\documentclass[10pt, conference, letter]{IEEEtran}
\usepackage{amsmath,amsfonts,amssymb,amsbsy, amsthm,ae,aecompl}
\usepackage{algorithm,algpseudocode}
\usepackage[utf8]{inputenc}
\usepackage[english]{babel}
\usepackage{bm}
\usepackage{color}
\usepackage[noadjust]{cite}
\usepackage{epsfig}
\usepackage{hhline}
\usepackage{caption}
\usepackage{booktabs}
\usepackage{enumerate}
\usepackage{mathtools}
\usepackage{cuted}
\usepackage{float}
\usepackage{fancyhdr}
\usepackage[T1]{fontenc}
\usepackage[acronym,toc,shortcuts]{glossaries}
\usepackage{graphicx, caption, subcaption}
\usepackage{lipsum}
\usepackage{dsfont}
\usepackage{transparent}
\usepackage{tikz,pgfplots}
\usepackage{times}
\usepackage{verbatim}
\usepackage[all]{xy}

\usepackage{etoolbox}
\AtBeginEnvironment{align}{\setcounter{subeqn}{0}}
\newcounter{subeqn} %

\usetikzlibrary{calc}
\usetikzlibrary{arrows,decorations.markings}
\pgfplotsset{
  grid style = {
    dash pattern = on 0.025mm off 0.95mm on 0.025mm off 0mm, 
    line cap = round,
    black,
    line width = 0.5pt
  },
  tick label style={font=\small},
  label style={font=\small},
  legend style={font=\footnotesize},
}

\pgfplotscreateplotcyclelist{laneas-energysto}{
cyan!60!black,		densely dotted, 	every mark/.append style={fill=cyan!80!black},mark=otimes*\\%
lime!80!black,		densely dashed,	every mark/.append style={fill=lime},mark=otimes*\\%
red,                solid, every mark/.append style={fill=red!80!black},mark=otimes*\\%
cyan!60!black, 		densely dotted,	every mark/.append style={fill=cyan!80!black},mark=diamond*\\%
lime!80!black, 		densely dashed,	every mark/.append style={fill=lime},mark=diamond*\\%
red,           		solid, every mark/.append style={solid,fill=red!80!black},mark=diamond*\\%
teal,				every mark/.append style={fill=teal!80!black},mark=*\\%
orange,				every mark/.append style={fill=orange!80!black},mark=square*\\%
red!70!white,		mark=star\\%
yellow!60!black,		densely dashed, every mark/.append style={solid,fill=yellow!80!black},mark=square*\\%
black,				every mark/.append style={solid,fill=gray},mark=otimes*\\%
blue,				densely dashed,mark=star,every mark/.append style=solid\\%
red,					densely dashed,every mark/.append style={solid,fill=red!80!black},mark=diamond*\\%
}

\graphicspath{{figures/}}

\bgroup
\def\arraystretch{1.7}%

\makeglossaries
\newacronym{ADMM}{ADMM}{Alternating Direction Method of Multipliers}
\newacronym{APC}{APC}{area power consumption}
\newacronym{AEC}{AEC}{average energy consumption}

\newacronym{ASE}{ASE}{area spectral efficiency}
\newacronym{AR}{AR}{average data rate}
\newacronym{BS}{BS}{base station}
\newacronym{CDN}{CDN}{content delivery network}
\newacronym{CN}{CN}{core network}
\newacronym{ICN}{ICN}{information-centric network}
\newacronym{CF}{CF}{collaborative filtering}
\newacronym{CRP}{CRP}{{C}hinese restaurant process}
\newacronym{CS}{CS}{central scheduler}
\newacronym{D2D}{D2D}{device-to-device}

\newacronym{ICIC}{ICIC}{inter-cell interference coordination}
\newacronym{LTE}{LTE}{long term evolution}

\newacronym{SISO}{SISO}{single-input single-output}
\newacronym{SBS}{SBS}{small base station}
\newacronym{SINR}{SINR}{signal-to-interference-plus-noise ratio}
\newacronym{SIR}{SIR}{signal-to-interference ratio}
\newacronym{SCN}{SCN}{small cell network}
\newacronym{SVD}{SVD}{singular value decomposition}
\newacronym{UT}{UT}{user terminal}
\newacronym{QoS}{QoS}{quality-of-service}
\newacronym{QoE}{QoE}{quality-of-experience}
\newacronym{RAN}{RAN}{radio access network}

\newacronym{PGFL}{PGFL}{probability generating functional}
\newacronym{HetNet}{HetNet}{heterogeneous network}

\newacronym{PPP}{PPP}{Poisson point process}
\newacronym{SE}{SE}{spectral efficiency}

\newacronym{SDMA}{SDMA}{space-division multiple access }
\newacronym{PDF}{PDF}{probability density function}
\newacronym{ZF}{ZF}{zero-forcing}
\newacronym{MIMO}{MIMO}{multiple-input multiple-output}

\newacronym{UE}{UE}{user}
\newacronym{EE}{EE}{energy efficiency}
\newacronym{MRC}{MRC}{maximum ratio combining}
\newacronym{CSI}{CSI}{channel state information}
\newacronym{RF}{RF}{radio frequency}

\makeglossaries

\newtheorem{definition}{Definition}
\newtheorem{theorem}{Theorem}
\newtheorem{lemma}{Lemma}

\newtheorem{remark}{Remark}

\IEEEoverridecommandlockouts \IEEEpubid{\makebox[\columnwidth]{ 978-1-5386-3531-5/17/\$31.00~\copyright~2017 IEEE \hfill} \hspace{\columnsep}\makebox[\columnwidth]{ }}
\begin{document}
\title{Downlink Performance  of Dense Antenna Deployment: To Distribute or Concentrate?}
\author{
		\IEEEauthorblockN{Mounia Hamidouche$^{ \star}$, Ejder Baştuğ$^{\dagger \diamond }$, Jihong Park$^{\circ}$,  Laura Cottatellucci$^{ \star}$, and Mérouane Debbah$^{\diamond \bullet}$ }
		\IEEEauthorblockA{
				\vspace{-0.25cm}
				\\
				\small
				$^{\star}$Communication Systems Department, EURECOM, Campus SophiaTech, 06410, Biot, France \\
				$^{\dagger}$Research Laboratory of Electronics, Massachusetts Institute of Technology, 77 Massachusetts Avenue, Cambridge, MA 02139, USA \\
				$^{\diamond}$Large Networks and Systems Group (LANEAS), CentraleSupélec, 91192, Gif-sur-Yvette, France \\
				$^{\circ}$Department of Electronic Systems, Aalborg University, Denmark \\
				$^{\bullet}$Mathematical and Algorithmic Sciences Lab, Huawei France R\&D, France \\	
				\vspace{-0.2cm} \\
				\{mounia.hamidouche, laura.cottatellucci\}@eurecom.fr, ejder@mit.edu,   \\
				jihong@es.aau.dk, merouane.debbah@centralesupelec.fr \\ 	
				\vspace{-1.25cm}
		}
\thanks{This  research  has  been  supported  by  the  ERC Starting  Grant  305123  MORE  (Advanced  Mathematical  Tools  for ComplexNetwork Engineering), and the U.S. National Science Foundation under Grant CCF-140922.}
}
\IEEEoverridecommandlockouts
\maketitle

\begin{abstract} 
Massive multiple-input multiple-output (massive MIMO) and small cell densification are complementary key 5G enablers. Given a fixed  number of the entire base-station antennas per unit area, this paper fairly compares (i) to deploy few base stations (BSs) and concentrate many antennas on each of them, i.e. massive MIMO, and (ii) to deploy more BSs equipped with few antennas, i.e. small cell densification. We observe that small cell densification always outperforms for both signal-to-interference ratio (SIR) coverage and energy efficiency (EE), when each BS serves multiple users via $L$ number of sub-bands (multi-carrier transmission). Moreover, we also observe that larger $L$ increases SIR coverage while decreasing EE, thus urging the necessity of optimal 5G network design. These two observations are based on our novel closed-form SIR coverage probability derivation using stochastic geometry, also validated via numerical simulations.
\end{abstract}
\section{Introduction}
\label{sec:introduction}
Performance improvements such as having wider coverage, higher user data rate, higher energy efficiency and lower latency are under an intensive investigation for the design of the fifth generation (5G) and beyond cellular architectures. These improvements are fundamental to achieve the dramatic growth of connected devices and the tremendous amount of data in applications such as voice, videos, and games \cite{Cisco2017}, as well as applications in wireless virtual-reality \cite{Bastug2017Toward}.
Novel innovative network technologies are used to meet the required performances. First, transmission with massive \gls{MIMO}  \cite{hoydis2013massive} is considered as a candidate technology for 5G. The key feature of this technology is the use of a large number of antennas at the base station compared to the number of users. The more antennas the base stations are equipped with, the better the performance is in terms of data rate and energy consumption \cite{wong2002performance}. A second promising technology is small cell networks \cite{hoydis2011green}, that consists of a dense number of small cell base stations in a given area. Due to the short distance between the \gls{BS} and the user terminals, small cell networks have a low path loss thus yielding a low power consumption which can improve the \gls{EE}. For a given total number of \gls{BS} antennas, interesting strategies in the deployment of these two technologies are either \emph{i) low density deployment of base stations with many antennas, i.e, massive \gls{MIMO} } or \emph{ii) a higher density deployment of \gls{BS}s equipped with fewer antennas}. Understanding which strategy is preferable is one of the goals of this work.

In fact, many works analyzed the impact of the massive amount of antennas on the \gls{EE}. In particular, the work in \cite{bjornson2015optimal} solves the \gls{EE} maximization problem for a multi-cell multi-user \gls{MIMO} network and shows that small cells yield higher \gls{EE}. In \cite{bjornson2015optimal}, the authors give insights on how the number of antennas at the \gls{BS} must be chosen in order to uniformly cover a given area and attain maximal \gls{EE}. Altough many works study massive \gls{MIMO} and small-cell densification, very few have focused on comparing their performance. Our goal is to analyze which one of the two technologies perform better in terms of the coverage probability and energy efficiency. 

A comparison has been recently presented in \cite{nguyen2016massive}. The massive \gls{MIMO} and small-cell systems were compared in terms of spectral and energy efficiency bounds.  The authors observe via simulations that for the average spectral efficiency, small-cell densification is favourable in crowded areas with moderate to high user density and massive \gls{MIMO} is preferable in scenarios with low user density. In contrast to the analysis in \cite{nguyen2016massive}, we derive exact expressions of the coverage probability and \gls{EE} by assuming other constraints on the model then we compare between massive MIMO and small cell networks in terms of these two metrics. One of the intersting constraint is to assume multi-carrier transmission in which the total bandwidth is divided into $L$ $\geq 1$ sub-bands. Then, instead of studying the downlink performance when each \gls{BS}  serve a single user at each time/frequency, we consider that  each \gls{BS} is scheduled to serve simultaneously multiple users on each sub-band. We also cancel the interference by using \ac{ZF} processing.  In addition, instead of introducing massive \gls{MIMO} and small-cell systems separately, we examine the problem with a single system model by varying the number of \gls{BS} antennas under the constraint of a fixed total number of \gls{BS} antennas per unit area.  In this work, we derive analytic expressions for the coverage probability and \gls{EE} using a stochastic geometry approach \cite{baccelli2010stochastic}. The key feature of this approach is that the base station
positions are all independent which allows to use tools from stochastic geometry.

The rest of this paper is organized as follows. Section \ref{sec:nmodel} details our system model of downlink transmission using linear processing \gls{ZF} under perfect \gls{CSI} at each base station. General expressions for coverage probability and \gls{EE} are derived in Section \ref{sec:COVERAGE}. In Section \ref{sec:numerical}, numerical results are used to validate the theoretical analysis and make comparisons between massive \gls{MIMO} and small cell densification  for both coverage probability and \gls{EE} metrics. Finally, the major conclusions and implications are given in Section \ref{sec:conclusion}. 

The following notation is used in this paper. The expectation operation with respect to a random variable and the absolute value are denoted by $\mathbb{E}\lbrace.\rbrace$ and $\vert.\vert$, respectively. We denote by $\textbf{I}_{M}$ the $M \times M$ identity matrix, and we use $\mathcal{CN}(0,\bf{\Sigma})$ to denote a circularly symmetric complex Gaussian distribution with zero-mean and covariance matrix $\bf{\Sigma}$. The Gamma function is denoted as $\Gamma(.)$. The bold lower-case letters as $\textbf{h}$ represent vectors, whereas the bold upper-case as $\textbf{H}$ are matrices. 
\section{Network Model}
\label{sec:nmodel}
\begin{figure*}[t]
	\centering
	\includegraphics[width=0.95\linewidth]{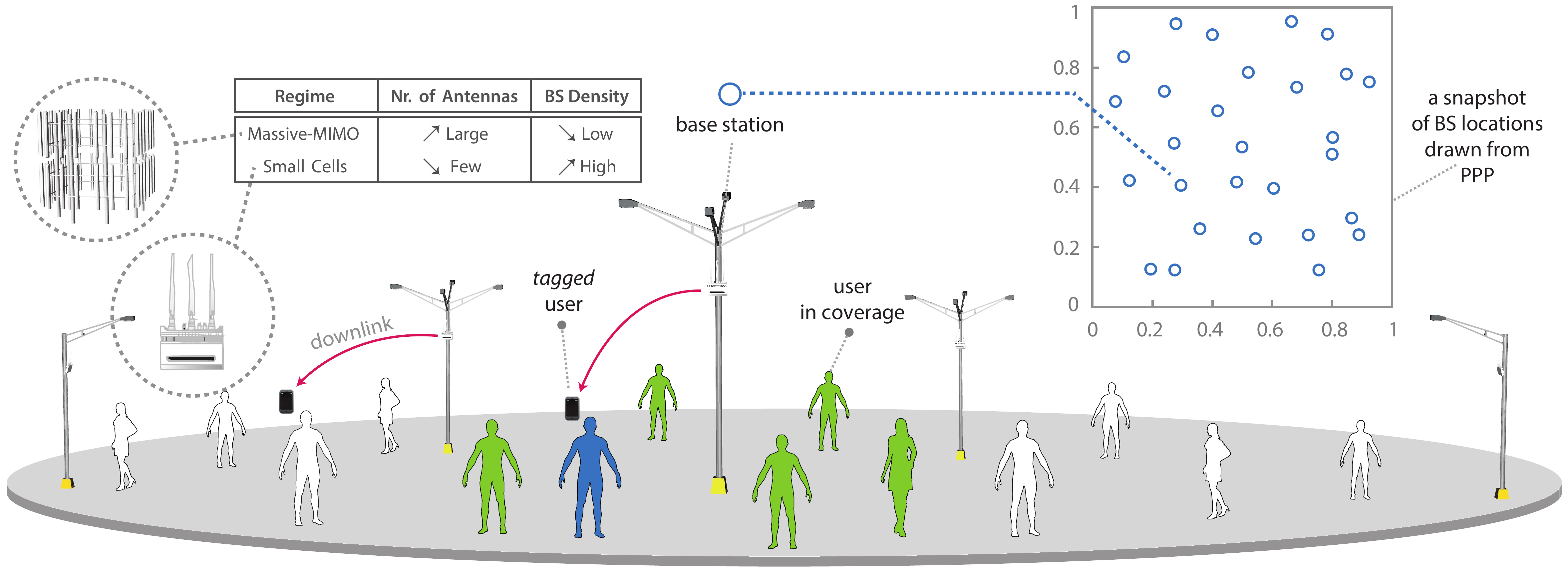}
	\caption{An illustration of the network model.}
	\label{fig:scenario}
\end{figure*} 
The cellular network consists of \gls{BS}s independently distributed according to a homogeneous \gls{PPP} $\Phi$ of intensity $\lambda_{\mathrm{BS}}$ (measured in \gls{BS}s/$\mathrm{km}^{2}$), and is depicted in Figure \ref{fig:scenario}. Each base station is equipped with an array of $M$ antennas. We consider an independent collection of single antenna mobile users, located according to another independent stationary \gls{PPP} $\Psi$ with intensity $\lambda_{\mathrm{UE}}$. We assume that each user connects to its closest \gls{BS}, namely each \gls{BS} serves the users which are located within its Voronoi cell \cite{haenggi2012stochastic}. In this section, assuming perfect \gls{CSI} at each \gls{BS} we study the signal model for downlink system. With this goal in mind, we first consider a typical user, which is connected to a tagged \gls{BS} ($\mathrm{BS}_{0}$). Since user locations are translation-invariant, we consider that the typical user is always located at the origin. This typical user's received signal $y$ is then given as
\begin{equation}\label{eq:downlink:y}
y=\underbrace{r_{0}^{-\alpha}\textbf{h}_{0}^{H} \textbf{x}_{0}}_{\mathrm{Desired \ signal} }+\underbrace{\sum\limits_{\substack{\mathrm{BS}_{i} \in \Phi \backslash \lbrace \mathrm{BS}_{0}\rbrace}}^{}{r_{i}^{-\alpha }\textbf{h}_{i}^{H} \textbf{x}_{i}}}_{\mathrm{Interference}}+n,
\end{equation}\normalsize
where the stochastic vector $\textbf{h}_{i} \in \mathbb{C}^{M}$ denotes the small scale fading between the $i$-th base station to the typical user. It follows a complex generalized Gaussian distribution denoted as $\textbf{h}_{i} \backsim \mathcal{C}\mathcal{N}(0,\textbf{I}_{M})$. The channel is considered to be noisy, with the Gaussian noise $n$ of variance $\frac{K}{P}$ added to the received signal. The variable $r_{i}$ is the distance from the typical user to its closest base station $\mathrm{BS}_{i}$ and $\alpha \geq 2$ is the path-loss exponent. We denote $\mathrm{\textbf{x}_{i}} \in \mathbb{C}^{M}$ an arbitrary symbol transmitted from the $i$-th base station. In addition to the received signal model in \eqref{eq:downlink:y}, we define the \gls{SINR} of the typical user as
\begin{equation}
\mathrm{SINR}=\dfrac{r_{0}^{-\alpha}\textbf{h}_{0}^{H} \textbf{x}_{0}}{\sum\limits_{\substack{\mathrm{BS}_{i} \in \Phi \backslash \lbrace \mathrm{BS}_{0}\rbrace}}^{}{r_{i}^{-\alpha }\textbf{h}_{i}^{H} \textbf{x}_{i}+n}}.
\end{equation}
We suppose that the \gls{BS}s must be deployed to match a given finite user density of $\lambda_{\mathrm{UE}}$ $\mathrm{UE}$s/$\mathrm{km}^{2}$, then each base station serves in average $\mathcal{K}=\frac{\lambda_{\mathrm{UE}}}{\lambda_{\mathrm{BS}}}$ users. The total bandwidth $W$ is divided into $L$ $\geq 1$ sub-bands. Therefore, $K=\frac{\mathcal{K}}{L} \leq M$ users are simultaneously served on each sub-band by each base station. We assume that the total number of antennas $\lambda_{\mathrm{BS}}M$ is fixed and should be deployed in a given area. Based on this assumption, for simplicity we set   $\lambda_{\mathrm{BS}}M$= $\lambda_{\mathrm{UE}}$, then the number of antennas in each base station is given by $M=\frac{\lambda_{\mathrm{UE}}}{\lambda_{\mathrm{BS}}}$ which is always equal or greater than $K$. We suppose that $P=\mathbb{E}[\mathrm{\textbf{x}}_{i}^{H}\mathrm{\textbf{x}}_{i}]$ is the average transmit power per base station which is given as $\frac{P_{\mathrm{max}}}{\lambda_{\mathrm{BS}}}$, where $P_{\mathrm{max}}$ is the maximum power used when all antennas are concentrated on a single base station. In this scenario, to cancel out the interference while boosting the desired signal power, each \gls{BS} applies \gls{ZF} transmission to simultaneously serve $K$ single antennas.
Let $s_{i,k} \backsim \mathcal{C}\mathcal{N}(0,1)$ be the message (the symbol) determined for the user $k$ from the $i$-th base station. Then, the $i$-th \gls{BS} multiplies the data symbol $s_{i,k}$ destined for the $k$-th user by $\mathrm{\textbf{w}_{i,k}}$. Therefore, the linear combination $\mathrm{ \textbf{x}_{i}}$ of the symbols transmitted by the $i$-th base station intended for the $K$ users is
\begin{equation}
\textbf{x}_{i}= \sum\limits_{\substack{k=1}}^{K}{\textbf{w}_{i,k} s_{i,k}},
\end{equation}
where $\textbf{w}_{i,k} \in \mathbb{C}^{M \times 1}$ is \gls{ZF} beamforming vector. Then, the received \gls{SINR} at the typical user can now be expressed as
\begin{equation}
\mathrm{SINR}=\dfrac{r_{0}^{-\alpha} \vert \textbf{h}_{0}^{H} \textbf{w}_{0,1}\vert^{2}}{\sum\limits_{\substack{\mathrm{BS}_{i} \in \Phi \backslash \lbrace \mathrm{BS}_{0}\rbrace}}^{}{r_{i}^{-\alpha}}
 \sum\limits_{\substack{k=1}}^{K}{\vert \textbf{h}_{i} \textbf{w}_{i,k}\vert^{2}}+\frac{K}{P}}= \dfrac{r_{0}^{-\alpha}S}{I_{r}+\frac{K}{P}},
\end{equation}
where $I_{r}=\sum\limits_{\substack{\mathrm{BS}_{i} \in \Phi \backslash \lbrace \mathrm{BS}_{0}\rbrace}}^{}{r_{i}^{-\alpha}}g_{i}$ with $g_{i}= \sum\limits_{\substack{k=1}}^{K}{\vert \textbf{h}_{i} \textbf{w}_{i,k}\vert^{2}}$ denotes the interference channel power and $S=\vert \textbf{h}_{0}^{H} \textbf{w}_{0,1}\vert^{2}$ is the desired channel power. Now, we introduce the metrics we will investigate in the next sections.
\begin{definition}[Coverage Probability]
The coverage probability of a typical user is the the probability that the \gls{SINR} received by the user is larger than a predefined threshold $T$ such as
\begin{equation}
\mathbb{P}_{\mathrm{cov}}(T)=\mathbb{P}(\mathrm{SINR} >T).
\end{equation}
\end{definition}
\begin{definition}[Energy Efficiency]
The \gls{EE} is defined as
\begin{equation}
\footnotesize
\mathrm{EE}= \dfrac{\mathrm{ASE}}{\mathrm{AEC}} = \dfrac{\mathrm{Area \ Spectral \ Efficiency} \;\mathrm{ [bit/symbol/km^{2}]}}{\mathrm{Average \ Energy \ Consumption} \;\mathrm{ [Joule/symbol/km^{2}]}},
\end{equation}
where the \gls{ASE} is expressed as
\begin{equation}\label{eq:ee:ase}
\mathrm{ASE}=\lambda_{\mathrm{BS}} K \mathbb{E}[R],
\end{equation}
in which $\lambda_{\mathrm{BS}}$ represents \gls{BS} density, $K$ is the number of users  that are served by the \glspl{BS} and $\mathbb{E}[R]$ is the average data rate of users. Moreover, \ac{AEC} is defined as similar to \cite{zhang2014energy}, that is
\begin{equation}\label{eq:ee:aec}
\mathrm{AEC}=\left( \dfrac{P}{\eta}+ MP_{\mathrm{c}}+K^3P_{\mathrm{pre}}+P_{\mathrm{0}}\right),
\end{equation}
where $\eta$ denotes the power amplifier efficiency, $P_{c}$ is the circuit power per antenna, which indicates the energy consumption of the corresponding RF chains. The term $K^3P_{\mathrm{pre}}$ accounts for the energy consumption for precoding which is related to the number of users served simultaneously by each base station. The term $P_{0}$ is the non-transmission power, which accounts for the energy consumption of baseband processing. 
\end{definition}
\begin{remark}
The goal in \gls{EE} is to maximize the performance while minimizing the energy consumption which are two conflicting operations in 5G.
\end{remark}
\section{Performance Analysis}
\label{sec:COVERAGE}
In this section we give the analytical expression of the coverage probability on a typical mobile user. Afterwards, we shall give the expression of the \gls{EE}. We start by stating the following Lemma that shall be used to derive our results.

\begin{lemma}[\cite{haenggi2012stochastic}] 
\label{Distance}
The probability density function (PDF) of a typical user's association distance $r_0$ is
\begin{equation}
f_{r_0}(r)=\frac{dF_{r_0}(r)}{dr}= e^{-\lambda_{\mathrm{BS}} \pi r^{2}}2\pi \lambda_{\mathrm{BS}} r.
\end{equation}

\begin{proof}
Each user is connected to its closest base station, then all the interfering base stations are farther than a distance $r$. Since the Poisson distribution helps in describing the chances of occurrence of a number of events in a given space, then the probability that no base station is closer than a distance $r$  within an area $\pi r^2$, is $e^{-\lambda_{\mathrm{BS}} \pi r^2}$, expressed as
\begin{align}
\hspace{-5pt} \mathbb{P}(r_0>r)=F_{r_0}(r) =\mathbb{P}[\mathrm{No \ BS} \ \mathrm{within}   \ \pi r^{2}]= e^{-\lambda_{\mathrm{BS}} \pi r^{2}}
\end{align}
 Therefore, the PDF results from the derivative of the cumulative distribution function $F_{r}(R)$.
\end{proof}
\label{dis}
\end{lemma}
\subsection{Coverage Probability}
Before deriving the expression of the coverage probability, we first derive  the Laplace transform of both interference and desired signal. The desired channel power $S$ is distributed as $\Gamma(M-K+1,1)$ \cite{dhillon2013downlink}.   For the interfering signal, as $\textbf{w}_{i,k}$ is a unit-norm vector and independent of $\textbf{h}_{i}$, then $\vert \textbf{h}_{i} \textbf{w}_{i,k}\vert^{2}$ is a squared-norm complex Gaussian, which is exponential distributed. For tractability we neglect the correlation between $\textbf{w}_{i,k}$ for different $k$, then the channel gain $g_{i}$ is the sum of $K$ independent exponential distributed random variables which follows $\Gamma(K,1)$. 
%
\begin{lemma}[Laplace Transform of Interference]
\label{the:LaplaceInterferanceDownlink}

\vspace{-0.2cm}

The Laplace transform of interference $\mathcal{L}_{I_{r}}(s)$=$\mathbb{E}[e^{-sI_{r}}]$ is 

\vspace{-5pt}\small\begin{align}
\mathcal{L}_{I_{r}}(s)= \exp \left[ -\pi \lambda_{\mathrm{BS}} r_{0}^{2} \left( -1 + \mathrm{_{2}F_{1}}\left( K, \dfrac{-2}{\alpha};\dfrac{-2}{\alpha}+1;-sr_{0}^{-\alpha} \right)\right)  \right],
\end{align}
where $_{2}\mathrm{F}_{1}$ is the Gauss-Hypergeometric function.
\end{lemma}\normalsize
\begin{proof}
See Appendix \ref{app:LaplaceInterferanceDownlink}.
\end{proof}

\begin{lemma}[Laplace Transform of the Desired Signal] 
\label{the:LaplaceDesiredDownlink}
The Laplace transform of the desired signal $\mathcal{L}_{S}(s)= \mathbb{E}[e^{-sS}]$ is 
\begin{equation}
\mathcal{L}_{S}(s)=\left( \dfrac{1}{1+s}\right)^{M-K+1}.
\end{equation}
\end{lemma}
\begin{proof}
See Appendix \ref{app:LaplaceDesiredDownlink}.
\end{proof}

Combining the previous results given in Lemmas  \ref{Distance}, \ref{the:LaplaceInterferanceDownlink} and \ref{the:LaplaceDesiredDownlink} with the proof techniques proposed in \cite{baccelli2009stochastic2}, an expression for the coverage probability can be derived and it is given in the following theorem.

\begin{theorem}[Coverage Probability in Downlink]
\label{the:covDownlink}
The coverage probability at a typical mobile user in the general cellular network model described above is

\begin{alignat}{5}
\begin{split}
\mathbb{P}_{\mathrm{cov}}(T)=\displaystyle{\int_{r_{0}>0}^{}}\displaystyle{\int_{-\infty}^{\infty}}\mathcal{L}_{I_{r_{0}}}(&i2\pi r^{\alpha}_{0}Ts) \exp\left( -\dfrac{i2\pi r_{0}^{\alpha}TK}{P}s\right)\\
& \times \dfrac{\mathcal{L}_{S}(-i2\pi s)-1}{i2\pi s}f_{r_{0}}(r_{0})\mathrm{d}s\mathrm{d}r_{0},
\end{split}
\end{alignat}
where $f_{r_{0}}(r_{0})$ is the PDF of the distance between the typical user and the tagged base station ($\mathrm{BS}_{0}$), $\mathcal{L}_{I_{r_{0}}}(.)$ is the Laplace transform of the interference and $\mathcal{L}_{S}(.)$ is the  Laplace transform of the desired signal (Lemmas \ref{Distance},  \ref{the:LaplaceInterferanceDownlink} and \ref{the:LaplaceDesiredDownlink} respectively).
\end{theorem}
\begin{proof}
See Appendix \ref{app:covDownlink}.
\end{proof}

\subsection{Energy Efficiency}
\begin{figure*}
\begin{equation}
\small
\mathrm{EE}= \dfrac{\lambda_{\mathrm{BS}} K \displaystyle{\int_{r_{0}>0}^{}}\displaystyle{\int_{-\infty}^{\infty}}\mathcal{L}_{I_{r_{0}}}(i2\pi r^{\alpha}_{0}Ts) \exp\left( -\dfrac{i2\pi r_{0}^{\alpha}TK}{P}s\right)\times \dfrac{\mathcal{L}_{S}(-i2\pi s)-1}{i2\pi s}f_{r_{0}}(r_{0})\mathrm{d}s\mathrm{d}r_{0}\log(1+T)}{\left( \dfrac{P}{\eta}+ MP_{\mathrm{c}}+K^3P_{\mathrm{pre}}+P_{\mathrm{0}}\right)}          
\label{EE}
\end{equation}
\hrule
\vspace{0.0cm}
\end{figure*}
\vspace{-0.2cm}
To facilitate the analysis of the \gls{EE}, we consider fixed modulation and coding schemes for each user by considering a fixed $\mathrm{SINR}$ threshold $T$ as in \cite{Park2017SeqCov}, providing the average rate $\mathbb{E}[R]$ as a function of downlink coverage probability as
\begin{alignat}{5}
\begin{split}
\mathbb{E}[R]=&\log(1+T) \mathbb{P}_{\mathrm{cov}}(T)\\
=&\log(1+T)\displaystyle{\int_{r_{0}>0}^{}}\displaystyle{\int_{-\infty}^{\infty}}\mathcal{L}_{I_{r_{0}}}(i2\pi r^{\alpha}_{0}Ts) \times\\
& \exp\left( -\dfrac{i2\pi r_{0}^{\alpha}TK}{P}s\right)\dfrac{\mathcal{L}_{S}(-i2\pi s)-1}{i2\pi s}f_{r_{0}}(r_{0})\mathrm{d}s\mathrm{d}r_{0},
\end{split}
\end{alignat}
By plugging the average rate expression into \eqref{eq:ee:ase}, we obtain the \gls{ASE}. Then, the expressions of \ac{EE} can be readily obtained and its final expression is given in \eqref{EE} on the top of this page. 
%

\section{Numerical Results}
\label{sec:numerical}
\bgroup
\def\arraystretch{1.6}%
\begin{table}[!t]
	\caption{\sc Default Simulation Parameters.}
	\label{tab:simparams}
	\centering
	\small
	\scalebox{0.90}{
		\begin{tabular}{|c|c|c|}
			\hline
			{\bf System Parameter}		& {\bf Symbol} 			& {\bf Value}\\
			\hline
			\hline
			 Power amplifier 			 	& $\eta$ 					& $0.318$  \\
			Circuit power per antenna 	& $P_{\mathrm{c}}$ 				& $14.8$ $\mathrm{W} $ 				\\
			Energy consumption for precoding            	& $P_{\mathrm{pre}}$ 			& $1.74$ $\mathrm{W}$  \\
			Nr. of Sub-bands & $L$  & $1$  \\
			Target $\mathrm{SINR}$	& $T$  & 1 $\mathrm{dBm}$   \\
			\hline
			\hline
			Non-transmission power & $P_{0}$  	&  $65.8$ $\mathrm{W}$  \\
			Maximum average power per BS                	& $P_{\mathrm{max}}$  	& $40$ $\mathrm{dBm}$ \\
			Users density   & $\lambda_{\mathrm{UE}}$  & $32$ per $\mathrm{km}^{2}$ \\
			\gls{BS} density	& $\lambda_{\mathrm{BS}}$  & $4$  \\
			Path-loss exponent & $\alpha$  & $4$  \\

			\hline
		\end{tabular}
	}
\end{table}

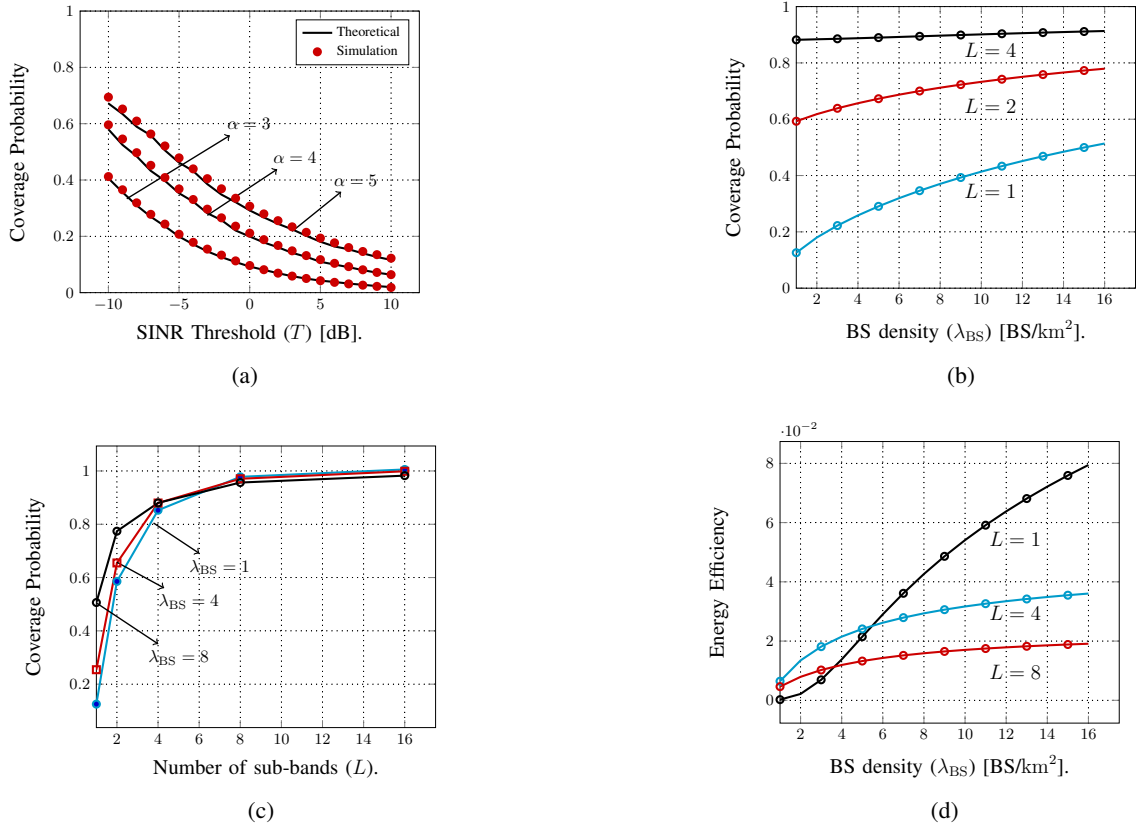
\begin{figure*}
	\centering
	\begin{subfigure}[b]{0.475\textwidth}   
		\centering 
		\hspace{-1cm}
		\scalebox{0.70}{
\begin{tikzpicture}[scale=0.94]
	\begin{axis}[
 		grid = major,
 		legend cell align=left,
 		mark repeat={1},
 		ymin=0,ymax=1,
 		legend style ={legend pos=north east},
 		xlabel={\large SINR Threshold ($T$) [dB].},
 		ylabel={\large Coverage Probability}]
 			\node at (axis cs:-0.8000, 0.5000) [anchor=south east] {};
 	\addplot+[black!60!black, solid, very thick, mark=none]
 		table [col sep=comma] {\string"././covProb-Alpha===3.csv"};
 		\addlegendentry{Theoretical} 
 		\addlegendentry{Simulation}
 		 
 	
 			\addplot+[red!80!black, solid, very thick, mark=none, only marks] 
 		table [col sep=comma] {\string"././covProb-Alpha===3-Sim.csv"};
 		\node[anchor=east] (source) at (axis cs:3.000000, 0.20000){};
 	\node (destination) at (axis cs: 7.50000, 0.40000){$\alpha = 5$};
       	\draw[thick,->](source)--(destination);
 
 		
 	\addplot+[black!60!black, solid, very thick, mark=none]
  table [col sep=comma] {\string"././covProb-Alpha===4.csv"};
 				  
 	  
	\addplot+[red!80!black, solid, very thick, mark=none, only marks] 
 		table [col sep=comma] {\string"././covProb-Alpha===4-Sim.csv"};
 		\node[anchor=east] (source) at (axis cs:-3.0000000, 0.2600000){};
       	\node (destination) at (axis cs:3.20000, 0.480){$\alpha = 4$};
       	\draw[thick,->](source)--(destination);
		 		
			\node at (axis cs:-2.200000, 0.05000) [anchor=south east] {};
 	\addplot+[black!60!black, solid, very thick, mark=none]
 				 table [col sep=comma] {\string"././covProb-Alpha===5.csv"};
 				 \node[anchor=east] (source) at (axis cs:-8.7000000, 0.3200){};
 				  	\node (destination) at (axis cs:-0.050000, 0.600){$\alpha = 3$};
       	\draw[thick,->](source)--(destination);

 		\addplot+[red!80!black, solid, very thick, mark=none, only marks]  
 		table [col sep=comma] {\string"././covProb-Alpha===5-Sim.csv"};	 

	\end{axis}
\end{tikzpicture}
		}
		\caption[]%
		{ }     
		\label{fig:coverageVsThreshold}
	\end{subfigure}
	\hfill
	\begin{subfigure}[b]{0.475\textwidth}
		\centering
		\hspace{-1cm}
		\scalebox{0.70}{
\begin{tikzpicture}[scale=0.94]
	\begin{axis}[
 		grid = major,
 		legend cell align=left,
 		mark repeat={2},
 		ymin=0,ymax=1,
 		xmin=1,
 		legend style ={legend pos=north east},
 		xlabel={\large BS density ($\lambda_{\mathrm{BS}}$) [BS/$\mathrm{km}^{2}$].},
 		ylabel={\large Coverage Probability}]
 			\node at (axis cs:12.000000, 0.300000) [anchor=south east] {\large $L = 1$};
 		\addplot+[cyan!80!black, solid, very thick, mark=o] table [col sep=comma] {\string"././covProb-T-L1.csv"};

 		\node at (axis cs:12.000000, 0.600000) [anchor=south east] {\large $L = 2$};
 		\addplot+[red!80!black, solid, very thick, mark=o, forget plot]
 				  table [col sep=comma] {\string"././covProb-T-L2.csv"};
		 		  
			\node at (axis cs:12.000000, 0.800000) [anchor=south east] {\large $L = 4$};
 		\addplot+[black!80!black, solid, very thick, mark=o, forget plot]
 				 table [col sep=comma] {\string"././covProb-T-L4.csv"};
	
	\end{axis}
\end{tikzpicture}
		}
		\caption[]%
		{ }    
		\label{fig:coverageVsDensity}
	\end{subfigure}	
	\vskip\baselineskip
	\begin{subfigure}[b]{0.475\textwidth}  
		\centering 
		\hspace{-1cm}
		\scalebox{0.70}{
\begin{tikzpicture}[scale=0.94]
	\begin{axis}[
 		grid = major,
 		legend cell align=left,
 		mark repeat={1},
 		xmin=1,
 		legend style ={legend pos=north east},
 		xlabel={\large Number of sub-bands ($L$).},
 		ylabel={\large Coverage Probability}]
 			\node at (axis cs:-2.20000, 0.5000) [anchor=south east] {\large $\alpha = 3$};
 		\addplot+[cyan!80!black, solid, very thick, mark=*] table [col sep=comma] {\string"././covProb-L-lambda=1.csv"};
 	
 			\node[anchor=east] (source) at (axis cs:3.7800000, 0.820000){};
       	\node (destination) at (axis cs:7.000000, 0.64000){$\lambda_{\mathrm{BS}}=1$};
       	\draw[thick,->](source)--(destination);
 		
 		\addplot+[red!80!black, solid, very thick, mark=square] 		
 				  table [col sep=comma] {\string"././covProb-L-lambda=4.csv"};
 				  \node[anchor=east] (source) at (axis cs:2.000000, 0.670000){};
       	\node (destination) at (axis cs:5.500000, 0.510000){$\lambda_{\mathrm{BS}}=4$};
       	\draw[thick,->](source)--(destination);
		 	  
			
 		\addplot+[black!80!black, solid, very thick, mark=o, forget plot]
 				 table [col sep=comma] {\string"././covProb-L-lambda=16.csv"};
 				 	  \node[anchor=east] (source) at (axis cs:1.00000, 0.520000){};
       	\node (destination) at (axis cs: 5.000000, 0.30000){$\lambda_{\mathrm{BS}
 				 	  }= 8$};
       	\draw[thick,->](source)--(destination);

	\end{axis}
\end{tikzpicture}
		}
		\caption[]%
		{ }
		\label{fig:coverageVsL}
	\end{subfigure}
	\quad
	\begin{subfigure}[b]{0.475\textwidth}   
		\centering 
		\hspace{-1cm}
		\scalebox{0.70}{
\begin{tikzpicture}[scale=0.94]
	\begin{axis}[
 		grid = major,
 		legend cell align=left,
 		mark repeat={2},
 		xmin=1,
 		legend style ={legend pos=north east},
 		xlabel={\large BS density ($\lambda_{\mathrm{BS}}$) [BS/$\mathrm{km}^{2}$].},
 		ylabel={\large Energy Efficiency}]
 			\node at (axis cs:14.000000, 0.05000) [anchor=south east] {\large $L = 1$};
 		\addplot+[black!80!black, solid, very thick, mark=o] table [col sep=comma] {\string"././EE-L1.csv"};

 		\node at (axis cs:14.000000, 0.02500) [anchor=south east] {\large $L = 4$};
 		\addplot+[cyan!80!black, solid, very thick, mark=o, forget plot]
 				  table [col sep=comma] {\string"././EE-L4.csv"};
		 		  
			\node at (axis cs:14.000000, 0.005000) [anchor=south east] {\large $L = 8$};
 		\addplot+[red!80!black, solid, very thick, mark=o, forget plot]
 				 table [col sep=comma] {\string"././EE-L8.csv"};
	
	\end{axis}
\end{tikzpicture}
		}
		\caption[]%
		{ }  
		\label{fig:EEVsDensity}
	\end{subfigure}
	\caption[]
	{Evolution of coverage probability with respect to the a) the target $\mathrm{SINR}$. b) the \gls{BS} density in \gls{BS}/$\mathrm{km}^{2}$, and c) number of sub-bands $L$. The subfigure d) shows the evolution of \gls{EE} with respect to the \gls{BS} density in \gls{BS}/$\mathrm{km}^{2}$.}
	\label{fig:numResults}
\end{figure*}
In this section, we conduct Monte-Carlo simulations to validate the  analytical expressions of coverage probability and  \gls{EE} of our multi-user \gls{MIMO} system. The default parameter setting is given in Table \ref{tab:simparams} and shall be used unless otherwise stated.

Figure \ref{fig:coverageVsThreshold} illustrates the coverage probability expression provided in Theorem \ref{the:covDownlink} as a function of target $\mathrm{SINR}$ for three different path-loss values: $\alpha \in \lbrace 3, 4, 5\rbrace$. The curves reveal that the coverage probability obtained by simulation behaves exactly as the analytical results which confirm the accuracy of our theoretical expressions.  Also, increasing $\alpha$ increases the coverage probability because the interference power decreases faster as a function of $\alpha$ than the power signal \cite{Andrews2011Tractable}.

Figure \ref{fig:coverageVsDensity}  shows the coverage probability as a function of \gls{BS} density. Several important observations can be made from the results of this figure. First, for any given $L$, the coverage probability can be greatly improved by increasing the \gls{BS} density, meaning that distributed network densification is preferable over massive \gls{MIMO}. Second, we note that the coverage probability increases as $L$ increases, thus showing the importance of multi-carrier transmissions. However, it is shown that for large $L$, both massive \gls{MIMO} and small cells provide the same coverage  which is confirmed according to Figure \ref{fig:coverageVsL}. This phenomenon occurs since the number of users $K$ served by each \gls{BS} decreases and become very small compared to the number of its own antennas $M$.

Figure \ref{fig:EEVsDensity} shows the impact of varying $\lambda_{\mathrm{BS}}$ on the \gls{EE}. The curves show that for any given $L$, the \gls{EE} can also be improved by increasing the \gls{BS} density, meaning that small cells is also preferable over massive \gls{MIMO}. In contrast with the coverage probability, we can observe that in small cells scenario, decreasing the number of sub-bands $L$ increases the \gls{EE} but in massive \gls{MIMO} scenario, large number of sub-bands provides the highest \gls{EE}. Finally, we notice also that both massive \gls{MIMO} and small cells provide the same performance gain when $L$ becomes very large which means that the \gls{EE} increase when the number of users served by the base station approach the number of its base station antennas $M$.  All these observations show that the coverage and \gls{EE} are conflicting such that improvements in one objective lead to degradation
in the other objective for a fixed number of sub-bands $L $.

\section{Conclusions}
\label{sec:conclusion}
Our proposed model was based on stochastic geometry, where the \gls{BS} and user locations were distributed according to \gls{PPP}. Using tools from stochastic geometry, we derived the coverage probability and \gls{EE} expressions for downlink scenario. The coverage probability expression was validated via Monte-Carlo simulations. The comparison showed that for any given sub-band $L$, small cell densification is preferable over massive \gls{MIMO} if both coverage probability and \gls{EE} should be increased. However, increasing $L$ improves the coverage probability but decreases the \gls{EE} which shows that the two metrics are conflicting. An interesting future work is  therefore to introduce a multi-objective optimization framework and  find (possibly) optimal number of sub-bands and \gls{BS} density that maximize the coverage probability and \gls{EE} jointly for the uplink and downlink scenarios. Future work will also look at multiple antenna terminals which is still an open topic. 
\bibliographystyle{IEEEtran}
\bibliography{references}
%
\appendices
\section{Proof of Lemma \ref{the:LaplaceInterferanceDownlink} }
\label{app:LaplaceInterferanceDownlink}
Let $f(g)$ and $f(S)$ denote the PDF of $g_{i}= \sum\limits_{\substack{k=1}}^{K}{\vert \mathrm{\textbf{h}}_{i} \mathrm{\textbf{w}}_{i,k}\vert^{2}}$ and $S=\vert \mathrm{\textbf{h}}_{0}^{H} \mathrm{\textbf{w}}_{0,1}\vert^{2}$ respectively. The Laplace transform of the interference is $\mathcal{L}_{I_{r}}(s)= \mathbb{E}[e^{-sI_{r}}]$, where the average is taken over both the spatial \gls{PPP} and the interference distribution is expressed as follows:
\begin{alignat}{5}
\begin{split}
\mathcal{L}_{I_{r}}(s)& = \mathbb{E}_{\Phi,\lbrace g_{i} \rbrace} \left[ \exp\left( -s \sum\limits_{\substack{\mathrm{BS}_{i}\in \Phi\backslash \lbrace \mathrm{BS}_{0}\rbrace}}^{}{r_{i}^{-\alpha}g_{i}}\right) \right],\\
& =\mathbb{E}_{\Phi, \lbrace g_{i} \rbrace}\left[ \underset{\mathrm{BS}_{i}\in \Phi \setminus\lbrace \mathrm{BS}_{0}\rbrace}{\prod} \left[ \exp(-sg_{i}r_{i}^{-\alpha})\right]\right],\\
 & \overset{(a)}{=}\mathbb{E}_{\Phi}\left[ \underset{\mathrm{BS}_{i}\in \Phi \setminus\lbrace \mathrm{BS}_{0}\rbrace}{\prod} \mathbb{E}_{g_{i}}\left[ \exp(-sg_{i}r_{i}^{-\alpha})\right]\right],\\
 & \overset{(b)}{=} \mathbb{E}_{\Phi} \left[ \underset{\mathrm{BS}_{i}\in \Phi \setminus\lbrace \mathrm{BS}_{0}\rbrace}{\prod} \frac{1}{\Gamma(K)} \displaystyle{\int_{0}^{\infty}}e^{-g_{i}(sr_{i}^{-\alpha}+1)}g_{i}^{K-1} \mathrm{d}g_{i}\right],\\
 &\overset{(d)}{=} \mathbb{E}_{\Phi} \left[ \underset{\mathrm{BS}_{i}\in \Phi \setminus\lbrace \mathrm{BS}_{0}\rbrace}{\prod} \frac{1}{(1+sr_{i}^{-\alpha})^{K}} \right],\\
 & \overset{}{=}\exp \left( -2\pi \lambda_{\mathrm{BS}} \displaystyle{\int_{r_{0}}^{\infty}} \left( 1-\frac{1}{(1+sv^{-\alpha})^{K}}\right)v\mathrm{d}v\right)
\end{split}
\end{alignat} 
where the step (a) follows from the i.i.d distribution of $g_{i}$ and further independence from the point process $\Phi$. The step (b) follows from the PDF of $g_{i}$ $\backsim$ $\Gamma(K,1)$ given as
\begin{equation}
 f(g)= \frac{1}{\Gamma(K)}g^{K-1} e^{-g_{i}}.
\end{equation}
Moreover, the step (c) follows from the computation of the integral by the means of integration by parts. The last step follows from the probability generating functional of the PPP with intensity $\lambda$ \cite{haenggi2012stochastic}, which states that for some function $f(x)$ we have
\begin{equation}
\mathbb{E}\left[ \underset{\mathrm{BS}_{i}\in \Phi \setminus\lbrace BS_{0}\rbrace}{\prod} f(x) \right] = \exp \left(-\lambda \displaystyle{\int_{\mathbb{R}^{2}}^{}}(1-f(x)\mathrm{d}x)\right) .
\end{equation}
The inside integral can be evaluated by using the change of variables $v^{-\alpha} \rightarrow y$ and we obtain the result. \qed
%
\section{Proof of Lemma \ref{the:LaplaceDesiredDownlink} }
\label{app:LaplaceDesiredDownlink}
The Laplace transform of the desired signal is $\mathcal{L}_{S}(s)= \mathbb{E}[e^{-sS}]$, where the average is taken over the desired signal distribution expressed as follows:
\begin{alignat}{3}
\begin{split}
\mathcal{L}_{S}(s)&=\mathbb{E}_{S}[e^{-sS}]\\
& \overset{}{=} \frac{1}{\Gamma(M-K+1)} \displaystyle{\int_{0}^{\infty}}S^{M-K} e^{-S(1+s)}\mathrm{d}S,\\
&\overset{a}{=}\left( \frac{1}{1+s}\right)^{M-K+1},
\end{split}
\end{alignat}
where the first step follows from the PDF of the desired signal
\begin{equation}
f(S)=\frac{1}{\Gamma(M-K+1)}S^{M-K}e^{-S}\mathrm{d}S.
\end{equation} 
The step (a) follows from the computation of the integral by the means of integration by parts. \qed
%
\section{Proof of Theorem \ref{the:covDownlink} }
\label{app:covDownlink}
The first part of the proof follows by conditioning on the nearest \gls{BS} being at a distance $r_{0}$ from the typical user. Then, the probability of coverage is
\begin{equation}
\mathbb{P}_{\mathrm{cov}}(T,\lambda_{\mathrm{BS}},\alpha)=\displaystyle{\int_{r_{0}>0}^{}}e^{-\pi \lambda_{\mathrm{BS}} r_{0}^{2}} \ \mathbb{P}(\frac{r_{0}^{-\alpha}S}{I_{r}+\frac{K}{P}}> T) \ 2\pi \lambda_{\mathrm{BS}} r_{0}\mathrm{d}r_{0}.
\label{eq1}
\end{equation}

To evaluate $\mathbb{P}(\frac{r_{0}^{-\alpha}S}{I_{r}+\frac{K}{P}}> T)$, we use the proof techniques proposed in \cite{Andrews2011Tractable} and \cite{baccelli2009stochastic2}. Some assumptions are required for this computation:
\begin{itemize}
	\item[A1)] The desired signal $S$ admits a square integrable density.
	\item[A2)] Either the interference $I_{r}$ or the noise admits a density which is square integrable.
\end{itemize}

The interference and the noise are independent, then the second assumption imply that $I_{r}+\frac{K}{P}$ admits a PDF $ f_{I_{r}+\frac{K}{P}}(y)$ that is square integrable. Therefore, the coverage probability is expressed as follows:
\begin{alignat}{5}
\begin{split}
\mathbb{P}\left(\mathrm{SINR} >T \right) & =\mathbb{P}\left( \dfrac{r^{-\alpha}S}{I_{r}+\frac{K}{P}} >T \right),\\
& = \mathbb{P}\left( I_{r}+\frac{K}{P}  < (T r_{0}^{\alpha})^{-1}S\right), \\
& = \mathbb{E}_{S}\left\{ \mathbb{P}\left( I_{r}+\frac{K}{P}  < (T r_{0}^{\alpha})^{-1}S\right)\right\}, \\
& \overset{(a)}{=} \mathbb{E}_{S} \left\{  \displaystyle{\int_{0}^{S(Tr_{0}^{\alpha})^{-1}}} f_{I_{r}+\frac{K}{P}}(y)\mathrm{d}y\right\},\\
& = \mathbb{E}_{S} \left\{ \displaystyle{\int_{-\infty}^{+\infty}} f_{I_{r}+\frac{K}{P}}(y) \textbf{1}_{\left[ 0 \leq y \leq S(Tr_{0}^{-\alpha})^{-1}\right]} \mathrm{d}y \right\}. 
\end{split}
\end{alignat}
where the step (a) above follows from the definition
\begin{equation}
\mathbb{P}(a \leq X \leq b)=\displaystyle{\int_{a}^{b}f(x)\mathrm{d}x},
\end{equation}
where $f$ is the density function of the variable $X$. Using the Plancheral-Parseval theorem \cite{bremaud2013mathematical}, we obtain
\small
\begin{alignat*}{2}
 &= \mathbb{E}_{S} \left\{ \displaystyle{\int_{-\infty}^{+\infty}} e^{-2\pi i sI_{r}} \ \exp\left( -2i\pi\frac{K}{P} s\right)\dfrac{e^{2i \pi y s(Tr_{0}^{\alpha})^{-1}}-e^{2i \pi ys \times 0}}{2i\pi s} \mathrm{d}s \right\},\\
& = \mathbb{E}_{S} \left\{ \displaystyle{\int_{-\infty}^{+\infty}} \mathcal{L}_{I_{r}}(-2i\pi Tr_{0}^{\alpha}s) \ \exp\left( -2i\pi\frac{K}{P}Tr_{0}^{\alpha} s\right)\frac{e^{2i \pi y s}-1}{2i\pi s} \mathrm{d}s \right\}.
\end{alignat*}
\normalsize
Using Fubini's theorem \cite{leon1993fubini}, and moving the expectation inside
\small
\begin{alignat*}{3}
& =  \displaystyle{\int_{-\infty}^{+\infty}} \mathbb{E}_{S} \left\{ \mathcal{L}_{I_{r}}(-2i\pi Tr_{0}^{\alpha} s) \ \exp\left( -2i\pi Tr_{0}^{\alpha} \frac{K}{P} s\right) \dfrac{e^{2i \pi y s}-1}{2i\pi s}  \right\} \mathrm{d}s,\\ 
&=  \displaystyle{\int_{-\infty}^{+\infty}} \mathcal{L}_{I_{r}}(-2i\pi Tr_{0}^{\alpha} s) \ \exp\left( -2i\pi Tr_{0}^{\alpha} \frac{K}{P} s\right)\dfrac{\mathcal{L}_{S} (e^{2i \pi y s})-1}{2i\pi s} \mathrm{d}s.
\label{eq2}
\end{alignat*}
\normalsize
Combining the last expression and (\ref{eq1}) gives the result stated in Theorem \ref{the:covDownlink}. \qed
\end{document}